%% file: bonchiEtAliiFinal.tex
\documentclass[copyright,creativecommons]{eptcs}

\providecommand{\event}{B. Klin, P. Soboci\'nski (Eds.): \\
6th Workshop on Structural Operational Semantics (SOS'09)}
\providecommand{\volume}{19}
\providecommand{\anno}{2010}
\providecommand{\firstpage}{46}
\providecommand{\eid}{4}

\usepackage{breakurl}
\usepackage{graphicx}

\usepackage{lscape}

\usepackage{epsfig}  

\usepackage{amsthm}
\usepackage{alltt}
\usepackage{amsfonts}
\usepackage{amssymb}
\usepackage{latexsym}
\usepackage{stmaryrd}
\usepackage[arrow,matrix,frame,curve,cmtip]{xy}\CompileMatrices
\usepackage{url}

\input{macro}
\title{On Barbs and Labels in Reactive Systems\thanks{Research
    partly supported by the EU within the FP6-IST IP 16004 \textsc{SEnSOria} (\emph{Software Engineering for
 Service-Oriented Overlay Computers}) and carried out during the first author's tenure of
an ERCIM ``Alain Bensoussan'' Fellowship Programme.}}

\author{Filippo Bonchi
\institute{Centrum voor Wiskunde en
Informatica, Amsterdam, The Netherlands}
\email{Filippo.Bonchi@cwi.nl}
\and
Fabio Gadducci
\institute{Dipartimento di Informatica,
Universit\`a di Pisa, Italy}
\email{fabio@di.unipi.it}
\and
Giacoma Valentina Monreale
\institute{Dipartimento di Informatica,
Universit\`a di Pisa, Italy}
\email{vale@di.unipi.it}
}

\def\titlerunning{On Barbs and Labels in Reactive Systems}
\def\authorrunning{F. Bonchi, F. Gadducci \& G. V. Monreale}

\begin{document}

\maketitle

\begin{abstract}
  Reactive systems (RSs) represent a meta-framework aimed at deriving
  behavioral congruences for those computational formalisms whose
  operational semantics is provided by reduction rules. RSs proved a
  flexible specification device, yet so far most of the efforts
  dealing with their behavioural semantics focused on idem pushouts
  (IPOs) and saturated (also known as dynamic) bisimulations. In this
  paper we introduce a novel, intermediate behavioural equivalence:
  \emph{L-bisimilarity}, which is able to recast both its IPO and
  saturated counterparts. The equivalence is parametric with respect
  to a set $L$ of RSs labels, and it is shown that under mild
  conditions on $L$ it is indeed a congruence. Furthermore,
  $L$-bisimilarity can also recast the notion of barbed semantics for
  RSs, proposed by the same authors in a previous paper. In order to
  provide a suitable test-bed, we instantiate our proposal by
  addressing the semantics of (asynchronous) CCS and of the calculus
  of mobile ambients.
\end{abstract}


\section{Introduction}
\label{sec:Intro}

\emph{Reactive systems} (RSs)~\cite{DBLP:conf/concur/LeiferM00} are an
abstract formalism for specifying the dynamics of a computational
device. Indeed, the usual specification technique is based on a
reduction system, comprising a set of possible states of the device
and a relation among them, representing the possible evolutions of the
device.
%
The relation is often given inductively, freely instantiating
relatively few rewriting rules: despite its ease of use, the main
drawback of reduction-based solutions is poor compositionality, since
the dynamic behaviour of arbitrary stand-alone terms can be
interpreted only by inserting them in appropriate contexts, where a
reduction may take place. The theoretical appeal of RSs is their
ability to distill labelled transition systems (LTSs), hence,
behavioural equivalences, for devices specified by a reduction system.

The idea underlying RSs is simple: whenever a device specified by
a term $C[P]$ (i.e., a sub-term $P$ inserted into a unary context
$C[-]$) may evolve to a state $Q$, the associated LTS has a transition
$P \tr{C[-]} Q$ (i.e., the state $P$ evolves into $Q$ with a label
$C[-]$).
%
%
If all contexts are admitted, the resulting semantics is called
saturated, and the standard bisimilarity on the derived LTS is a
congruence.
However, it is unfeasible to check the bisimulation game under all
contexts, and usually it suffices to consider a subset of contexts
that guarantees that the distilled behavioural semantics is a
congruence.  Such a set, the ``minimal'' contexts allowing a reduction
to occur, was identified in~\cite{DBLP:conf/concur/LeiferM00} by the
notion of \emph{relative pushout}: the resulting strong bisimilarity
is a congruence, even if it often does not coincide with the saturated
one.

Several attempts have been made to encode various specification
formalisms (Petri nets~\cite{RobinBGPN,Sassone05}, logic
programming~\cite{LICS2006}, etc.) as RSs, either hoping to recover
the standard observational equivalences, whenever such a behavioural
semantics exists (CCS~\cite{MIL:CAC}, pi-calculus~\cite{Mil:PPCA},
etc.), or trying to distill a meaningful new semantics. The results
are often not fully satisfactory: bisimilarity via minimal contexts is
usually too fine-grained; while saturated semantics are often too
coarse (the standard CCS strong bisimilarity is e.g. strictly included
in the saturated one). As for process calculi, the standard way out of
the empasse it to consider \emph{barbs}~\cite{RobinICALP92} (i.e.,
predicates on the states of a system) and barbed equivalences
(i.e., adding the check of such predicates in the bisimulation game).
The flexibility
of the definition allows for recasting a variety of observational,
bisimulation-based equivalences. Indeed, the methodological
contribution of~\cite{BGMFOSSACS09} is the introduction of suitable
notions of barbed saturated semantics for RSs.

In this paper we move one step further, and we propose a novel
behavioural equivalence for RSs, namely, $L$-bisimulation: a
flexible tool, parametric with respect to a set of minimal labels $L$.
Also in this case the idea is very simple, and it just asymmetrically
refines the standard bisimulation game. If the minimal LTS has
a transition $P \tr{C[-]} Q$, then a bisimilar $P'$ has to react via
a minimal transition $P' \tr{C[-]} Q'$, whenever $C[-] \in L$;
or it must ensure that $C[P']$ may evolve into $Q'$ (thus requiring no
minimality for $C[-]$ with respect to $P'$), otherwise.
The associated bisimilarity is  intermediate between the standard semantics
(i.e., minimal and saturated)
for RSs: indeed, it is able to recover both of them, by simply
varying the set $L$ and exploiting the so-called semi-saturated semantics.
It can be proved that, under mild closure conditions on the set $L$,
$L$-bisimilarity is a congruence; and moreover, it can be shown
that  barbed saturated semantics can be recast, as long as $L$
satisfies suitable barb-capturing properties.

With respect to barbed saturated semantics, $L$-bisimilarity admits a
streamlined definition, where state predicates play no role. It is
thus of simpler verification, and its introduction may have far
reaching consequences over the usability of the RS formalism. However,
as for any newly proposed semantics, its adequacy and ease of use have
to be tested against suitable case studies. We thus consider a recently
introduced, minimal context semantics for mobile ambients (MAs), as distilled
in~\cite{BoGaMo}; as well as two minimally labelled transition systems
for CCS and its asynchronous variant, reminiscent of those proposed in
\cite{bgk:bisimulation-graph-enc}. We show that in those cases, a set
$L$ of minimal labels can be identified, such that $L$-bisimilarity
precisely captures the standard semantics of the calculus at hand.

The paper is organized as follows. Section~2 recalls the basic notions
of RSs, while Section~3 and Section~4 perform the same for MAs and
(asynchronous) CCS, respectively. Section~5 presents the technical
core of the paper: the introduction of $L$-bisimilarity for RSs, the
proof that (under mild conditions on $L$) it is indeed a congruence,
and moreover its correspondence with barbed semantics.  Finally,
Section~6 and Section~7 prove that, suitably varying the set $L$, the
newly defined $L$-bisimilarity captures the standard equivalences for
MAs and for CCS and its asynchronous variant, respectively.


\section{Reactive Systems}\label{sec:Rea}

This section summarizes the main results concerning (the theory of)
reactive systems (RSs)~\cite{DBLP:conf/concur/LeiferM00}.  The
formalism aims at deriving labelled transition systems (LTSs) and
bisimulation congruences for a system specified by a reduction
semantics, and it is centered on the concepts of \emph{term},
\emph{context} and \emph{reduction rule}: contexts are arrows of a
category, terms are arrows having as domain $0$, a special object that
denotes groundness, and reduction rules are pairs of (ground) terms.

\begin{definition}[Reactive System] A \emph{reactive system} $\mathbb{C}$
  consists of
  \begin{enumerate}
    \setlength{\itemsep}{0cm}
  \item a category $\mathbf{C}$;
  \item a distinguished object $0 \in |\mathbf{C}|$;
  \item a composition-reflecting subcategory $\mathbf{D}$ of
    \emph{reactive contexts};
  \item a set of pairs $\Rules{R}\subseteq \bigcup_{I\in |\mathbf{C}|}
    \mathbf{C}(0,I)\times
    \mathbf{C}(0,I)$ of \emph{reduction rules}.
  \end{enumerate}
\end{definition}

Intuitively, reactive contexts are those in which a reduction may
occur. By composition-reflecting we mean that $d'\circ d \in
\mathbf{D}$ implies $d,d'\in \mathbf{D}$.  Note that the rules have to
be ground, i.e., left-hand and right-hand sides have to be terms
without holes and, moreover, with the same codomain.

The reduction relation is generated from the reduction rules by
closing them under all reactive contexts.  Formally, the
\emph{reduction relation} is defined by taking $P \react Q$ if there
is $\langle l,r \rangle \in\Rules{R}$ and $d\in \mathbf{D}$ such that
$P=d\circ l$ and $Q=d \circ r$.

Thus the behaviour of an RS is expressed as an unlabelled transition
system.
In order to obtain a LTS, we can plug a term $P$ into some context $C[-]$ and observe if a reduction occurs. In this case we have that $P\tr{C[-]}$. Categorically speaking, this means that $C[-] \circ P$ matches $d \circ l$ for some rule $ \langle l,r \rangle\in \Rules{R}$ and some reactive context $d$. This situation is formally depicted by diagram (i) in Fig.~\ref{figureRedex}: a commuting diagram like this is called a \emph{redex square}.

\begin{definition}[Saturated Transition System]
  The \emph{saturated transition system} (STS) is defined as
  follows \begin{itemize} \item states: arrows $P:0\rightarrow I$ in
  $\mathbf{C}$, for arbitrary $I$; \item transitions: $P\FULLtr{C[-]}
  Q$ if $C[P] \react Q$.  \end{itemize}
\end{definition}

Note that $C[P]$ stands for $C[-] \circ P$: the same notation is
used in Definitions~\ref{defSB} and~\ref{defSSB} below, in order to
allow for an easier comparison with the process calculi notation, to
be adopted in the following sections.


\begin{definition}[Saturated Bisimulation]\label{defSB}
A symmetric relation $\mathcal R$ is a \emph{saturated
bisimulation} if whenever $P\, \mathcal R\,Q$ then $\forall C[-]$
\begin{itemize}
\item if $C[P]\react P'$ then $C[Q]\react Q'$ and $P'\,\mathcal R\,Q'$.
\end{itemize}
\emph{Saturated bisimilarity} $\sim^S$ is the largest saturated bisimulation.
\end{definition}

It is obvious that $\sim^S$ is a congruence. Indeed, it is the
coarsest symmetric relation satisfying the bisimulation game on
$\react$ that is also a congruence.

Note that STS is often infinite-branching since all contexts
allowing reductions may occur as labels.  Moreover, it has intuitively
redundant
transitions.  For example, consider the term $a.0$ of CCS. We have
both the transitions $a.0 \FULLtr{\overline{a}.0 \mid -} 0|0$ and
$a.0 \FULLtr{P\mid \overline{a}.0 \mid -} P\mid 0\mid 0$, yet $P$
does not ``concur'' to the reduction. We thus need a notion of
``minimal context allowing a reduction'', captured by \emph{idem
pushouts}.

\begin{definition}[RPO, IPO]\label{RPO}
  Let the diagrams in Fig.~\ref{figureRedex} be in a category
  $\mathbf{C}$, and let ($i$) be a commuting diagram.  A
  \emph{candidate} for (i) is any tuple $ \langle I_5,e,f,g \rangle$
  making (ii) commute.
  A \emph{relative pushout (RPO)} is the smallest such candidate,
  i.e., such that for any other candidate $\langle I_6, e',f',g'
  \rangle$ there exists a unique morphism $h:I_5 \rightarrow I_6$
  making (iii) and (iv) commute.
%
%
  A commuting square such as diagram (i) of Fig.~\ref{figureRedex} is
  called \emph{idem pushout (IPO)} if $\langle
  I_4,c,d,id_{I_4}\rangle$ is its RPO.
\end{definition}
%

 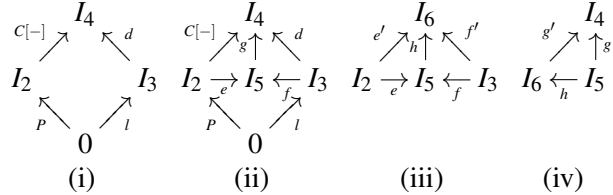
\begin{figure}[t]
 \begin{center}
   \begin{tabular}{@{}c@{}c@{}c@{}c}

   \def\labelstyle{\scriptscriptstyle}
 \xymatrix@C=10pt@R=10pt
     {
       &I_4 & \\
       I_2\ar[ur]^{C[-]}&    &I_3\ar[ul]_d\\
       &0\ar[ul]^P\ar[ur]_l &
     }

     &
  \def\labelstyle{\scriptscriptstyle}
     \xymatrix@C=10pt@R=10pt
     {
       &I_4 & \\       I_2\ar[ur]^{C[-]} \ar[r]_e& I_5 \ar[u]^g   &I_3\ar[l]^f \ar[ul]_d\\
       &0\ar[ul]^P\ar[ur]_l &
     }

     &

     \def\labelstyle{\scriptscriptstyle}
     \xymatrix@C=10pt@R=10pt
     {
       &I_6 & \\
       I_2\ar[ur]^{e'} \ar[r]_e& I_5 \ar[u]^h   &I_3\ar[l]^f \ar[ul]_{f'}
     }
     &
     \def\labelstyle{\scriptscriptstyle}
     \xymatrix@C=10pt@R=10pt
     {
       &I_4\\
       I_6 \ar[ur]^{g'}&I_5\ar[u]_g \ar[l]^h
     }
 \\
 (i) &(ii) & (iii) & (iv)

   \end{tabular}
\caption{Redex Square and RPO}\label{figureRedex}
\end{center}
\end{figure}

Hereafter, we say that an RS \emph{has redex RPOs} (\emph{IPOs})
if every redex square has an RPO (IPO) as candidate.
For a better understanding of these two
notions, we refer the reader to \cite{LICS2006}.
%

\begin{definition}[IPO Transition System]
The \emph{IPO transition system}
  (\ilts) is defined as follows
  \begin{itemize}
  \item states: $P:0\rightarrow I$ in $\mathbf{C}$, for
    arbitrary $I$;
  \item transitions: $P\IPOtr{C[-]} d\circ r$ if
    $d \in \mathbf{D}$,
    $\langle l,r \rangle \in \Rules{R}$,
    and (i) in Fig.~\ref{figureRedex} is an IPO.
  \end{itemize}
\end{definition}

In other words, if inserting $P$ into the context $C[-]$ matches
$d\circ l$, and $C[-]$ is the ``smallest'' such context,
then $P$ evolves to $d\circ r$ with label $C[-]$.

Bisimilarity on $\ilts$ is referred to as \emph{IPO-bisimilarity}
($\sim^{I}$).  Leifer and Milner have shown that if the RS has redex
RPOs, then it is a congruence.

\begin{proposition}\label{theoBiggest}
Let us consider an RS with redex RPOs. Then, $\sim^{I}$ is a congruence.
\end{proposition}

Clearly, $\sim^I\subseteq \sim^S$. In \cite{Bonchithesis} the first
author shows that this inclusion is strict for many formalisms. In
particular, it turns out that in some interesting cases $\sim^I$ is
too strict, while $\sim^S$ is too coarse. This fact is the reason for
introducing \emph{barbed bisimilarities} \cite{RobinICALP92}. Barbs
are predicates (representing some basic observations) on the states of
a system. For instance, in \cite{RobinICALP92} the authors use for CCS
barbs $\barb{a}$ and $\barb{\bar{a}}$ representing the ability of a
process to perform an input, respectively an output, on channel $a$.

In the following we fix a family $O$ of barbs, and we write
$P\barb{o}$ if $P$ satisfies $o \in O$.

\begin{definition}[Barbed Saturated Bisimulation]\label{def:BSbis}
A symmetric relation $\mathcal R$ is a \emph{barbed saturated
bisimulation} if whenever $P\, \mathcal R\,Q$ then $\forall C[-]$
\begin{itemize}
\item if $C[P]\barb{o}$ then $C[Q]\barb{o}$;
\item if $C[P]\react P'$ then $C[Q]\react Q'$ and $P'\, \mathcal R\,Q'$.
\end{itemize}
\emph{Barbed saturated bisimilarity} $\bsbis$ is the largest barbed
saturated bisimulation.
\end{definition}

It is easy to see that $\bsbis$ is the largest barbed bisimulation
that is also a congruence.

\subsection{An Efficient Characterization of (Barbed) Saturated Bisimilarity}

Since the definition of saturated bisimulation involves a
quantification over all possible contexts, it is usually hard to
(automatically) prove the equivalence of two systems. For this
reason, the first author, with K\"onig and Montanari, introduced
\emph{semi-saturated bisimilarity} \cite{LICS2006}.

\begin{definition}[Semi-Saturated Bisimulation]\label{defSSB}
A symmetric relation $\mathcal R$ is a \emph{semi-saturated
  bisimulation} if whenever $P\,\mathcal R\,Q$ then
\begin{itemize}
\item if $P\IPOtr{C[-]} P'$ then $C[Q]\react Q'$ and $P'\,\mathcal R\,Q'$.
\end{itemize}
\emph{Semi-saturated bisimilarity} $\sim^{SS}$ is the largest barbed
semi-saturated bisimulation.
\end{definition}

\begin{proposition}
  \label{theoSemiSatisSat}
Let us consider an RS with redex IPOs. Then,
 $\sim^{SS}=\sim^{S}$.
\end{proposition}
Reasoning on $\sim^{SS}$ is easier than on $\sim^{S}$ because
instead of looking at the reductions in all contexts,
only IPO transitions are considered.

%
%
%
%
%

In \cite{BGMFOSSACS09}, the authors extended this technique to barbed
saturated bisimilarity.

\begin{definition}[Barbed Semi-Saturated Bisimulation]\label{def:BSSbis}
  A symmetric relation $\mathcal R$ is a \emph{barbed semi-saturated
  bisimulation} if whenever $P\,\mathcal R\,Q$ then
\begin{itemize}
\item $\forall C[-]$, if $C[P]\barb{o}$ then $C[Q]\barb{o}$;
\item if $P\IPOtr{C[-]} P'$ then $C[Q]\react Q'$ and $P'\,\mathcal R\,Q'$.
\end{itemize}
\emph{Barbed semi-saturated bisimilarity} $\bssbis$ is the largest
barbed semi-saturated bisimulation.
\end{definition}


\begin{proposition} \label{PropBSS=BS}
Let us consider an RS with redex IPOs. Then,
 $\bssbis=\bsbis$.
\end{proposition}

Also in this case, it is more convenient to work with $\bssbis$
instead of $\bsbis$. Even if barbs are still quantified over all
contexts, for many formalisms (as for MAs) it is actually enough to
check if $P\barb{o}$ implies $Q\barb{o}$, since this condition
implies that $\forall C[-]$, if $C[P]\barb{o}$ then $C[Q]\barb{o}$.
Barbs satisfying this property are called \emph{contextual} barbs.

\begin{definition}[Contextual Barbs] \label{contBarb}
A barb $o$ is a \emph{contextual barb} if whenever
$P\barb{o}$ implies $Q\barb{o}$ then $\forall C[-]$, $C[P]\barb{o}$
implies $C[Q]\barb{o}$.
\end{definition}


\section{Mobile Ambients} \label{sec:MA}

In this section we first introduce the finite, communication-free fragment
of mobile ambients (MAs)~\cite{CG:MAmb} and its reduction semantics.
Then, we recall the IPO transition system for MAs presented in \cite{BoGaMo}.

Fig.~\ref{syntax} shows the syntax of the calculus. We assume a set
$\mathcal N$ of \emph{names} ranged over by $m, n, u, \ldots$.
Besides the standard constructors, we include a set $\{X, Y, \ldots\}$
of \emph{process variables} and a set $\{x, y, \dots\}$ of \emph{name
  variables}.
We let $P,Q,R, \ldots$ range over the set of \emph{pure} processes,
containing neither process nor name variables; while $P_\epsilon,
Q_\epsilon, R_\epsilon, \ldots$ range over the set of
\emph{well-formed} processes, i.e., such that no process or ambient
variable occurs twice.

Intuitively, an impure process such as $x[P] | X$ represents an
underspecified system, where either the process $X$ or the name
of the ambient $x[-]$ can be further instantiated. These extended
processes are needed later for the presentation of the LTS.
We use the standard definitions for the set of free names of a pure
process $P$, denoted by $fn (P)$, and for $\alpha$-convertibility,
with respect to the restriction operators $(\nu n)$. We moreover
assume that $fn(X)=\emptyset$ and $fn(x[P])= fn(P)$.
We also consider a family of \emph{substitutions}, which may replace a
process/name variable with a pure process/name, respectively.
Substitutions avoid name capture: for a pure process $P$, the
expression
$(\nu n)(\nu m)(X | x[\nil])\lbrace ^m/_x , ^{n[P]}/_X \rbrace$
corresponds to the pure process
$(\nu p)(\nu q)(n[P] | m[\nil])$,
for names $p, q \not \in \{m\} \cup fn(n[P])$.

The semantics of the calculus
exploits a \emph{structural congruence}, denoted by $\equiv$, which is
the least equivalence on pure processes that satisfies the axioms in
Fig.~\ref{cong}.
%
The \emph{reduction relation}, denoted by $\react$,
describes the evolution of pure processes.
It is the smallest relation closed under the congruence $\equiv$ and
inductively generated by the set of axioms and inference rules in Fig. \ref{reduction}.

As already said, a \emph{barb} $o$ is a predicate over the states of a system,
with $P\barb{o}$ denoting that $P$ satisfies $o$.  In MAs, $P\barb{n}$
denotes the presence at top-level of an unrestricted ambient
$n$. Formally, for a pure process $P$, $P\barb{n}$ if $P \equiv (\nu
A)(n[Q]|R)$ and $n \not \in A$, for processes $Q$ and $R$ and a
set of restricted names $A$.

\begin{figure}[!t]
\center
\begin{tabular}{p{7.3 cm} p{5 cm}}
\hline \\
$P::= \nil , n[P] , M.P , (\nu n)P , P_1 | P_2 , X , x[P]$ & $M::= in$ $n , out$ $n , open$ $n$ \bigskip\\
\hline
\end{tabular}
\caption{(Extended) Syntax of mobile ambients.}
\label{syntax}
\vspace{3mm}
\begin{tabular}{p{6 cm} p{6 cm}}
\hline
if $P \equiv Q$ then $P|R \equiv Q|R$ & $P|\nil \equiv P$\\
if $P \equiv Q$ then $(\nu n)P \equiv (\nu n)Q$ & $(\nu n)(\nu m)P \equiv (\nu m)(\nu n)P$\\
if $P \equiv Q$ then $n[P] \equiv n[Q]$ & $(\nu n)(P|Q) \equiv P|(\nu n)Q$ \hspace{3mm} if $n \notin fn(P)$\\
if $P \equiv Q$ then $M.P \equiv M.Q$ & $(\nu n) m[P] \equiv m[(\nu n)P]$ \hspace{3mm} if $n \ne m$\\
$P|Q \equiv Q|P$ & $(\nu n) M.P \equiv M.(\nu n)P$ \hspace{4mm} if $n \notin fn(M)$\\
$(P|Q)|R \equiv P|(Q|R)$ & $(\nu n) P \equiv (\nu m)(P\lbrace ^m/_n \rbrace)$ \hspace{.75mm}if $m \notin fn(P)$\\
\hline
\end{tabular}
\caption{Structural congruence.}
\label{cong}
\vspace{3mm}
\begin{tabular}{p{6 cm} p{6 cm}}
\hline
$n[in \hspace{1 mm} m.P|Q]|m[R] \react m[n[P|Q]|R]$ & if  $P \react Q$ then  $(\nu n)P \react (\nu n)Q$\\
$m[n[out \hspace{1 mm} m.P|Q]|R]\react n[P|Q]|m[R]$ & if  $P \react Q$  then  $n[P] \react n[Q]$ \\
$open \hspace{1 mm} n.P|n[Q] \react P|Q$ & if  $P \react Q$  then  $P|R \react Q|R$\\
\hline
\end{tabular}
\caption{Reduction relation on pure processes.}
\label{reduction}
\end{figure}

\begin{definition} [Reduction Barbed Congruences~\cite{PawelAmbient}]
\label{StrongCong}
\emph{Reduction barbed congruence} $\sim^{MA}$ is the largest symmetric relation $\mathcal R$ such that whenever
$P \, \mathcal R \, Q$ then
\begin{itemize}
\item if $P\barb{n}$ then $Q\barb{n}$;
\item if $P \react P'$ then $Q \react Q'$ and $P' \, \mathcal R \, Q'$;
\item $\forall C[-], C[P] \, \mathcal R \, C[Q]$.
\end{itemize}
\end{definition}

A labelled characterization of the reduction barbed congruence over
MAs processes was first presented by Rathke and Soboci\'nski in
\cite{PawelAmbient}, and then by the authors of this paper in
\cite{BGMFOSSACS09}.  In the latter we exploit the notion of barbed
saturated bisimilarity over RSs that is able to capture the
behavioural semantics for MAs defined above.  Moreover, we give an
efficient characterization of such a semantics through the IPO
transition system presented in the next section, and the
semi-saturated game.

\paragraph{An ITS for Mobile Ambients.}
Here we present the $\ilts$ $M$ for MAs proposed in
\cite{BoGaMo}. The inference rules describing this LTS
are obtained
from an analysis of a LTS over (processes as) graphs, derived by the
borrowed context mechanism \cite{EK:06}, which is an instance of the
theory of RSs~\cite{graphslics}.  The labels of the
transitions are unary contexts, i.e., terms of the extended syntax
with a hole $-$. Note that they are minimal contexts, that is, they
represent the exact amount of context needed by a system to react.  We
denote them by $C_\epsilon[-]$.  The formal definition of the LTS is
presented in Fig.~\ref{TableLTSProc2}.

The rule \textsc{Tau} represents the $\tau$-actions
modeling internal computations.  Notice that the labels of the
transitions are identity contexts composed of just a hole $-$, while the
resulting states are pure processes.

The other rules in Fig.~\ref{TableLTSProc2} model the interactions of
a process with its environment. Note that both labels and resulting
states contain process and name variables.  We define the LTS
$M_I$ for processes over the standard syntax of MAs by
instantiating all the variables of the labels and of the resulting
states.

\begin{definition}
\label{instance}
Let $P,Q$ be pure processes and let $C[-]$ be a pure context.  Then,
we have that $P \tr{C[-]}_{M_I} Q$ if there exists a
transition $P \tr{C_\epsilon[-]}_{M} Q_\epsilon$ and a
substitution $\sigma$ such that $Q_\epsilon \sigma \equiv Q$ and
$C_\epsilon[-] \sigma = C[-]$.
\end{definition}

In the above definition recall that substitutions replace process
variables by pure processes and that they do not capture bound names.

The rule \textsc{Open} models the opening of an ambient provided by
the environment. In particular, it enables a process $P$ with a
capability $open \hspace{1mm} n$ at top level, for $n \in fn(P)$, to
interact with a context providing an ambient $n$ containing some
process $X_1$. Note that the label $-|n[X_1]$ of the rule represents
the minimal context needed by the process $P$
for reacting.  The resulting state is the process over the extended
syntax $(\nu A)(P_1|X_1|P_2)$, where $X_1$ represents a process
provided by the environment.  Note that the instantiation of the
process variable $X_1$ with a process containing a free name that
belongs to the bound names in $A$ is possible only $\alpha$-converting
the resulting process $(\nu A)(P_1|X_1|P_2)$ into a process that does
not contain that name among its bound names at top level.

The rule \textsc{CoOpen} instead models an environment that opens an
ambient of the process.  The rule \textsc{InAmb} enables an ambient of
the process to migrate into a sibling ambient provided by the
environment, while in the rule \textsc{In} both ambients are provided
by the environment. In the rule \textsc{CoIn} an ambient provided by
the environment enters an ambient of the process. The rule
\textsc{OutAmb} models an ambient of the process exiting from an
ambient provided by the environment, while in the rule \textsc{Out}
both ambients are provided by the environment.

\begin{figure}[!t]
\begin{center}
\begin{tabular}{p{6.5 cm} p{6 cm}}
\hline \\
\textsc{\scriptsize(Tau)}
$\frac{P \react Q} {P \tr{-} Q}$
& \textsc{\scriptsize (Out)}
$\frac{P \equiv (\nu A) (out \hspace{1mm} m.P_1|P_2)
\hspace{2mm} m \not \in A}
      {P \tr{m[x[-|X_1]|X_2]} (\nu A) (m[X_2]| x[P_1|P_2|X_1])}$
\\ \\
\textsc{\scriptsize(In)}
$\frac{P \equiv (\nu A) (in \hspace{1mm} m.P_1|P_2)
\hspace{2mm} m \not \in A}
      {P \tr{x[-|X_1] | m[X_2]} (\nu A) m[x[P_1|P_2|X_1]|X_2]}$
& \textsc{\scriptsize(OutAmb)} $\frac{P \equiv (\nu A) (n[out \hspace{1mm}
m.P_1|P_2]|P_3) \hspace{2mm} m \not \in A}
         {P \tr{m[-|X_1]} (\nu A) (m[P_3|X_1]| n[P_1|P_2])}$
\\ \\
\textsc{\scriptsize (InAmb)}
$\frac{P \equiv (\nu A) (n[in \hspace{1mm}
m.P_1|P_2] |P_3) \hspace{2mm} m \not \in A}
      {P \tr{-|m[X_1]} (\nu A) (m[n[P_1|P_2]|X_1] | P_3)}$
& \textsc{\scriptsize(Open)} $\frac{P \equiv (\nu A) (open \hspace{1mm} n.P_1
|P_2) \hspace{2mm} n \not \in A}
      {P \tr{-|n[X_1]} (\nu A) (P_1|P_2|X_1)}$
\\ \\
\textsc{\scriptsize (CoIn)}
$\frac{P \equiv (\nu A) (m[P_1]|P_2) \hspace{2mm} m
\not \in A}
         {P \tr{-|x[in \hspace{1mm} m.X_1|X_2]} (\nu A) (m[x[X_1|X_2]|P_1]|P_2)}$
& \textsc{\scriptsize (CoOpen)} $\frac{P \equiv (\nu A) (n[P_1]|P_2) \hspace{2mm}
n \not \in A}
         {P \tr{-|open \hspace{1mm} n.X_1} (\nu A) (P_1|X_1| P_2)}$\\ \\
\hline \\
\end{tabular} \\
\caption{The LTS $M$.} \label{TableLTSProc2}
\end{center}
\end{figure}






\section{On Synchronous and Asynchronous CCS}\label{sec:CCS}
This section introduces the ITSs for CCS and for its asynchronous
variant.  For the sake of space, we do not present the standard CCS,
while we indeed recall the syntax and the semantics of Asynchronous
CCS (ACCS). We then show an ITS for both CCS and ACCS: the former was
introduced in~\cite{bgk:bisimulation-graph-enc}, while the latter is
original. Finally, we show that the IPO-bisimilarity coincides with
the ordinary bisimilarity for CCS; while IPO-bisimilarity is strictly
contained in asynchronous bisimilarity.

\paragraph{Asynchronous CCS.}
Differently from synchronous calculi, where messages are
simultaneously sent and received, in asynchronous communication the
messages are sent and travel through some media until they reach
destination. Thus sending is non blocking (i.e., a process may send
even if the receiver is not ready to receive), while receiving is
(processes must wait until a message becomes available).
Observations reflect the asymmetry: since sending is non blocking,
receiving is unobservable.

Here we shortly introduce the finite fragment of ACCS.
We adopt a presentation reminiscent of
asynchronous $\pi$~\cite{AmadioCS98} that allows the non
deterministic choice for input prefixes (a feature missing
in~\cite{BorealeNP98,CastellaniH98}).

Fig.~\ref{syntaxA} shows the syntax of the calculus. We assume a set
$\mathcal N$ of \emph{names} ranged over by $a, b, c, \ldots$.
As for MAs, we included a set $\{X, Y, \ldots\}$ of \emph{process
  variables}. These are needed for the presentation of the LTS in
Fig.~\ref{LTSACCS}.
We let $P,Q,R, \ldots$ range over the set of \emph{pure} processes,
containing no process variables. Substitution of process variables is
defined analogously to MAs. Note that here we let $M,N,O, \ldots$
range over the set of summation, while in MAs we used
those symbols for capabilities.

The main difference with respect to the standard CCS \cite{MIL:CAC} is the absence of
output prefixes. The occurrence of an unguarded $\bar{a}$ indicates
a message that is available on some communication media named $a$,
and it disappears whenever it is received.

\emph{Structural equivalence} ($\equiv$) is the smallest congruence
induced by the axioms in Fig.~\ref{congA}.
The behaviour of a process $P$ is then described as a reaction
relation ($\react$) over processes up to $\equiv$, obtained by
closing the rules in Fig.~\ref{reductionA}.
For ACCS, there exists also an interactive semantics
expressed by an LTS. This is the transition relation over processes
up to $\equiv$, obtained by the rules in Fig.~\ref{LTSAs}. Here we
use $\mu$ to range over the set of labels $\{\tau,a, \bar{a}\; |\;
a\in \mathcal N\}$. The names of $\mu$, denoted by $n(\mu)$, are
defined as usual.

\begin{figure}[!t]
\center
\begin{tabular}{p{7.3 cm} p{5 cm}}
\hline \\
$P ::= M,\, X,\, \ \bar{a},\, \ (\nu a)P,\, \ P_1 | P_2$   & $M ::= \nil,\, \ \tau.P,\, \ a.P,\, \  M_1 + M_2$ \bigskip\\
\hline
\end{tabular}
\caption{(Extended) Syntax of Asynchronous CCS.} \label{syntaxA}
\vspace{3mm}
\begin{tabular}{p{6 cm} p{6 cm}}
\hline
if $P \equiv Q$ then $P|R \equiv Q|R$ & $P|\nil \equiv P$\\
if $P \equiv Q$ then $(\nu a)P \equiv (\nu a)Q$ & $(\nu a)(\nu b)P \equiv (\nu b)(\nu a)P$\\
if $P \equiv Q$ then $\tau.P \equiv \tau.Q$ & $(\nu a)(P|Q) \equiv P|(\nu a)Q$ \hspace{3mm} if $a \notin fn(P)$\\
if $P \equiv Q$ then $a.P \equiv a.Q$ & $M+N \equiv N+M$ \\
if $M \equiv N$ then $M+O \equiv N+O$ & $(M+N)+O \equiv M+(N+O)$\\
$P|Q \equiv Q|P$ & $M+\nil \equiv M$ \\
$(P|Q)|R \equiv P|(Q|R)$ &  $(\nu a) P \equiv (\nu b)(P\lbrace ^b/_a \rbrace)$ \hspace{.75mm}if $b \notin fn(P)$\\
\hline
\end{tabular}
\caption{Structural congruence.}
\label{congA}
\vspace{3mm}
\begin{tabular}{p{6 cm} p{6 cm}}
\hline
$(a.P + M) | \bar{a} \react P$ & if  $P \react Q$ then  $(\nu a)P \react (\nu a)Q$\\
$\tau.P + M \react P$ & if  $P \react Q$  then  $P|R \react Q|R$ \\
\hline
\end{tabular}
\caption{Reduction relation on pure processes.}
\label{reductionA}
\vspace{3mm}
\begin{tabular}{p{3 cm} p{9 cm}}
\hline
$a.P +M \tr{a} P$ & if  $P \tr{\mu} Q$ then  $(\nu a)P \tr{\mu} (\nu a)Q$ \hspace{3mm} if $a \notin n(\mu)$\\
$\tau.P +M \tr{\tau} P$ & if  $P \tr{\mu} Q$  then  $P|R \tr{\mu} Q|R$ \\
$\bar{a} \tr{\bar{a}} \nil$ & if  $P \tr{a} P_1$ and $Q \tr{\bar{a}} Q_1$ then  $P|Q \tr{\tau} P_1|Q_1$ \\
\hline
\end{tabular}
\caption{Labelled transition system.} \label{LTSAs}
\end{figure}

The main difference with respect to the synchronous calculus lies in the notion
of \emph{observation}. Since sending messages is non-blocking, an
external observer can just send messages to a system without knowing
if they will be received or not. For this reason receiving should
not be observable and thus barbs take into account only outputs.
Formally, $P\downarrow \bar{a}$ if there exists process $Q$ such
that $P\tr{\bar{a}} Q$. This is reflected in the notion of
asynchronous bisimilarity~\cite{AmadioCS98}.

\begin{definition}[Asynchronous Bisimulation]\label{def:Abis}
A symmetric relation $\mathcal R$ is an \emph{asynchronous
bisimulation} if whenever $P\, \mathcal R\,Q$ then

\begin{itemize}
\item if $P\tr{\tau}P'$ then $Q\tr{\tau}Q'$ and $P'\, \mathcal
R\,Q'$,
\item if $P\tr{\bar{a}}P'$ then $Q\tr{\bar{a}}Q'$ and $P'\, \mathcal
R\,Q'$,
\item if $P\tr{a}P'$ then either $Q\tr{a}Q'$ and $P'\, \mathcal
R\,Q'$ or $Q\tr{\tau}Q'$ and $P'\, \mathcal R \, Q' | \bar{a}$.
\end{itemize}
\emph{Asynchronous bisimilarity} $\sim^A$ is the largest
asynchronous bisimulation.
\end{definition}

For example, the processes $a.\bar{a} + \tau.\nil$ and $\tau.\nil$
are asynchronous bisimilar. If $a.\bar{a} + \tau.\nil
\tr{a}\bar{a}$, then $\tau.\nil\tr{\tau}\nil$ and clearly $\bar{a}
\sim^A \nil |\bar{a}$.

\paragraph{An ITS for CCS.}
In \cite{bgk:bisimulation-graph-enc}, the first and the second
author together with K\"onig derived an ITS for the
ordinary CCS by employing the borrowed context mechanism
\cite{EK:06}.

Fig.~\ref{LTSCCS} shows the LTS $C$. The labels of
$C$ are minimal contexts, i.e., they represent the exact
amount of context needed by a process to react. The reactive
semantics of CCS (denoted by $\react$) can be found in
\cite{Mil:PPCA}. Note that both the labels and the resulting states
contain the process variable $X_1$. For the sake of space, we avoided to
report here the (extended) syntax of CCS: this is just the ordinary
syntax of CCS, together with process variables
(analogously to MAs and ACCS).

Following Definition~\ref{instance} for MAs, we define the LTS $C_I$ for processes
over the standard syntax by instantiating the process variable
of the labels and of the resulting states.

Now let us consider the rule \textsc{Rcv}. If a process is ready
to receive on some unrestricted channel $a$, then an interaction
takes place whenever it is embedded in an environment of the shape
$-|\bar{a}.X_1$\footnote{The LTS derived in
\cite{bgk:bisimulation-graph-enc} slightly differs from
$C$. Besides dropping some not-engaged transitions (i.e.,
transitions that do not play any role in the notion of
bisimulation), we simplified the labels for \textsc{Snd} and \textsc{Rcv}:
these were, respectively, $-|\bar{a}.X_1 + M_1$ and $-|\bar{a}.X_1+M_1$ for
$M_1$ a summation variable. Since these variables do not occur in
the resulting states, they also play no role in the derived bisimilarity, and thus
we avoided to consider them in the labels.}. Recall that
the instantiation of the process variable $X_1$ with a process
containing a free name that belongs to the bound names in $A$ is
possible only $\alpha$-converting the resulting process $(\nu
A)(Q|R|X_1)$.

Hereafter we use $\tr{\mu}$ (with $\mu \in \{\tau,a, \bar{a}\; |\;
a\in \mathcal N\}$) to denote the ordinary LTS of CCS
\cite{MIL:CAC}. By comparing the latter with the LTS $C$, it
is easy to see that $P\tr{\tau}Q$ if and only if $P\tr{-}Q$.
Moreover $P\tr{a}Q$ iff $P \tr{-|\bar{a}.X_1} Q|X_1$ and
$P\tr{\bar{a}}Q$ iff $P \tr{-|a.X_1} Q|X_1$. From these
facts, the main result of \cite{bgk:bisimulation-graph-enc}
follows: the ordinary bisimilarity of CCS (denoted by $\sim^{CCS}$)
coincides with IPO bisimilarity. Instead, saturated bisimilarity is
too coarse: the (recursive) processes
$P= rec_z\tau. z$ and $P | a.\nil$ are e.g. saturated bisimilar.

\begin{figure}[!t]
\begin{center}
\begin{tabular}{p{6.5 cm} | p{6 cm}}
\begin{tabular}{l}
\textsc{\scriptsize(Tau)} $\frac{P \react Q}{P \tr{-} Q}$\\\\
\textsc{\scriptsize(Rcv)} $\frac{P \equiv (\nu A) (a.Q+M |
R)\hspace{2mm} a \not \in A}{P \tr{-|\bar{a}.X_1} (\nu A)
(Q|R|X_1)}$\\\\
\textsc{\scriptsize(Snd)} $\frac{P \equiv (\nu A) (\bar{a}.Q +M |
R)\hspace{2mm} a \not \in A}{P \tr{-|a.X_1} (\nu A) (Q|R|X_1)}$\\\\
\end{tabular}
\caption{The LTS $C$}\label{LTSCCS} &
\begin{tabular}{l}
\textsc{\scriptsize(Tau)} $\frac{P \react Q} {P \tr{-} Q}$\\\\
\textsc{\scriptsize(Rcv)} $\frac{P \equiv (\nu A) (a.Q+M |
R)\hspace{2mm} a \not \in A}{P \tr{-|\bar{a}} (\nu A) (Q|R)}$\\\\
\textsc{\scriptsize(Snd)} $\frac{P \equiv (\nu A) (\bar{a} |
Q)\hspace{2mm} a \not \in A}{P \tr{-|a.X_1} (\nu A) (Q|X_1)}$\\\\
\end{tabular}
\caption{The LTS $A$}\label{LTSACCS}
\end{tabular}
\end{center}
\end{figure}

\paragraph{An ITS for ACCS.} Following \cite{bgk:bisimulation-graph-enc},
we propose an ITS for ACCS.
Fig.~\ref{LTSACCS} shows the LTS $A$. The LTS
$A_I$ is defined by instantiating the process variable of
the labels and of the resulting states.

The main difference between $A$ and $C$ is in
the rule \textsc{Rcv}: since outputs have no
continuation in ACCS, then the process variable $X_1$ (that occurs
in $C$) is not needed in $A$.

It is easy to see that also for ACCS there is a
close correspondence between the ordinary LTS semantics (in
Fig.~\ref{LTSAs}) and $A$: $P\tr{\tau}Q$ iff $P\tr{-}Q$,
$P\tr{a}Q$ iff $P \tr{-|\bar{a}} Q$ and $P\tr{\bar{a}}Q$ iff $P
\tr{-|a.X_1} Q|X_1$.

However, in the asynchronous case, IPO-bisimilarity is too fine
grained. Indeed, the processes $a.\bar{a} + \tau.\nil$ and
$\tau.\nil$ are asynchronously bisimilar, but they are not
IPO-bisimilar. In the next section we will introduce a new semantics
for RSs that generalizes both $\sim^{CCS}$ and
$\sim^A$.


\section{A New Semantics for Reactive Systems: $L$-Bisimilarity} \label{sec:L1-Bisimulation}

As shown in Section \ref{sec:CCS}, IPO-bisimilarity coincides with the
ordinary bisimilarity in the case of CCS.  However, for many
interesting calculi, such as MAs and ACCS, it is often too
fine-grained. On the other side, as recalled above for CCS, saturated
bisimilarity is often too coarse.

In this section we introduce $L$-indexed bisimilarity (shortly,
$L$-bisimilarity), a novel kind of bisimilarity parametric with
respect to a class of contexts (also referred to as \emph{labels})
$L$. For each class $L$ satisfying some closure properties, the new
equivalence $\bisl$ is a congruence and $\sim^I \subseteq \bisl
\subseteq \sim^S$.

Intuitively, $L$-bisimulations can be thought of as something in between
IPO-bisimulations and semi-saturated bisimulations: if $C[-]$ belongs
to $L$, then $Q$ must perform $Q\IPOtr{C[-]}$ whenever $P\IPOtr{C[-]}$
(as in the IPO-bisimulation), otherwise $C[Q]\react$ (as in the
semi-saturated bisimulation).

\begin{definition}[$L$-Bisimulation]\label{def:L1bis}
Let $L$ be a class of contexts.
A symmetric relation $\mathcal R$ is an $L$-\emph{bisimulation} if
whenever $P\, \mathcal R\,Q$ then

\[
\hbox{if $P\IPOtr{C[-]}P'$ then}
\left\{%
\begin{array}{ll}
    Q\IPOtr{C[-]}Q' \hbox{ and } P'\, \mathcal R\,Q', & \hbox{if $C[-]\in L$;} \\
    C[Q]\react Q' \hbox{ and } P'\, \mathcal R\,Q', & \hbox{otherwise.} \\
\end{array}%
\right.
\]
$L$-\emph{bisimilarity} $\bisl$ is the largest $L$-bisimulation.
\end{definition}

It is easy to note that $\bisl$ generalizes both $\sim^I$ and
$\sim^{SS}$ (and thus $\sim^S$). Indeed, in order to characterize
the former, it is enough to take as $L$ the whole class of contexts,
while to characterize the latter, we take as $L$ the empty class. In
Section~\ref{L=BS} we will show that for some $L$,
$L$-bisimilarity also coincides with barbed saturated bisimilarity.
In the remainder of this section, we show that $\bisl$ is a
congruence. In order to prove this, we have to require the following
condition on $L$.

\begin{definition} Let $L$ be a class of arrows of a category. We say that $L$ is
IPO-closed, if whenever the following diagram is an IPO and $b\in
L$, then also $c\in L$.
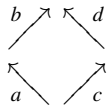
$$
\xymatrix@C=15pt@R=15pt
      {
        &\\
\ar[ur]^b & & \ar[ul]_d\\
& \ar[ul]^a \ar[ur]_c
      }$$
\end{definition}

It is often hard to prove that a class of contexts is IPO-closed.
It becomes easier with concrete instances of RSs
that supply a constructive definition for IPOs, such as
bigraphs and borrowed contexts.

\begin{proposition}\label{prop:BISLisacongruence}

Let us consider a RS with redex RPOs and an IPO-closed
class $L$ of contexts.
Then,  $\bisl$ is a congruence.
\end{proposition}
  \begin{center}
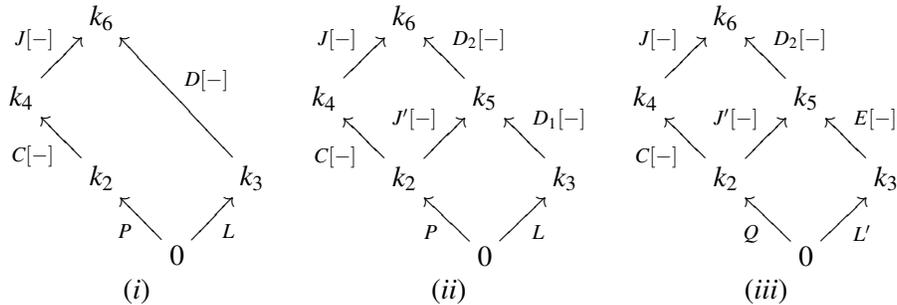
    \begin{tabular}{ccc}
      \xymatrix@C=15pt@R=15pt
      {
        &k_6\\
        k_4 \ar[ur]^{J[-]} & &\\
        &k_2\ar[ul]^{C[-]}  & &k_3 \ar[uull]_{D[-]}\\
        & &  0\ar[ul]^{P} \ar[ur]_{L}
      }

      &

      \xymatrix@C=15pt@R=15pt
     {
        &k_6\\
        k_4 \ar[ur]^{J[-]} & &k_5 \ar[ul]_{D_2[-]}\\
        &k_2\ar[ul]^{C[-]} \ar[ur]^{J'[-]} & &k_3 \ar[ul]_{D_1[-]}\\
        & &  0\ar[ul]^{P} \ar[ur]_{L}
      }

      &

      \xymatrix@C=15pt@R=15pt
     {
        &k_6\\
        k_4 \ar[ur]^{J[-]} & &k_5 \ar[ul]_{D_2[-]}\\
        &k_2\ar[ul]^{C[-]} \ar[ur]^{J'[-]} & &k_3 \ar[ul]_{E[-]}\\
        & &  0\ar[ul]^{Q} \ar[ur]_{L'}
      }
      \\
      ($i$)&($ii$)&($iii$)
    \end{tabular}
  \end{center}
\begin{proof}

  In order to prove this theorem we will use the composition and decomposition properties
  of IPOs, namely Proposition 2.1 and Proposition 2.2 of \cite{DBLP:conf/concur/LeiferM00}.
  We have to prove that if $P\bisl Q$ then $C[P]\bisl C[Q]$.
We show that $\mathcal R=\{(C[P],C[Q])$ s.t. $P\bisl Q
\}$ is an $L$-bisimulation.


  Suppose that $C[P]\IPOtr{J[-]}P'$. Then there exists
  an IPO square like diagram (i) above, where
 $\<L,R\>\in \Rules{R}$, $D[-]\in \Cat{D}$
  and $P'=D[R]$. Since, by hypothesis, the RS has
  redex RPOs, then we can construct an RPO as the one in diagram
  (ii) above. In this diagram, the lower square is an IPO, since RPOs
  are also IPOs (Proposition 1 of
  \cite{DBLP:conf/concur/LeiferM00}). Since the outer square is an IPO and the
  lower square is an IPO, by IPO decomposition property, it follows that
 also the upper square is an IPO.

Since $\Cat{D}$ is composition-reflecting, then both $D_1[-]$ and
$D_2[-]$ belong to $\Cat{D}$, and then $P\IPOtr{J'[-]}D_1[R]$. Now there
are two cases: either $J[-]\in L$ or $J[-]\notin L$.

If $J[-]\in L$, then also $J'[-]\in L$, because $L$ is IPO-closed,
by hypothesis. Since $P \bisl Q$, then $Q\IPOtr{J'[-]}Q''$ and $D_1[R]
\bisl Q''$. This means that there exists an IPO square like the
lower square of diagram (iii) above, where $\<L',R'\>\in \Rules{R}$,
$E[-]\in \Cat{D}$ and $E[R]=Q''$. Now recall by the previous
observation that the upper square of diagram (iii) is also an IPO
and then, by IPO composition, also the outer square is an IPO. This
means that $C[Q]\IPOtr{J[-]}D_2[Q'']$. Since $D_1[R] \bisl Q''$, then
$P'=D[R]=D_2[D_1[R]] \;\mathcal R\; D_2[Q'']$.

If $J[-]\notin L$, then either $J'[-]\in L$ or $J'[-]\notin L$. In both
cases, from $P\IPOtr{J'[-]}D_1[R]$ we derive that $J'[Q]\react Q''$ and
$D_1[R] \bisl Q''$. This means that the lower square of diagram
(iii) above commutes. Since also the upper square commutes, then
also the outer square commutes. This means that $C[Q]\react
D_2[Q'']$. Since $D_1[R] \bisl Q''$, then $P'=D[R]=D_2[D_1[R]] \;
\mathcal R \; D_2[Q'']$.
\end{proof}

\subsection{Barbed Saturated Bisimilarity via $L$-bisimilarity}\label{L=BS}

Here we show that $L$-bisimilarity can also characterize barbed
saturated bisimilarity, whenever the family of barbs and the set of
labels $L$ satisfy suitable conditions. This result will be used in
later sections in order to show that $L$-bisimilarity captures the
correct equivalences for MAs and ACCS.

In order to guarantee that $\bisl \subseteq \bsbis$, we need some
conditions ensuring that the checking of barbs of $\bsbis$ is
already done in $\bisl$ by the labels in $L$.
\begin{definition}\label{def:PredCapturing}
Let $L$ be a set of labels and let $O$ be a set of barbs. We say
that $L$ is \emph{$O$-capturing} if for each barb $o$ there exists a
label $C[-] \in L$ such that for each process $P$, $P \barb o$ if and only if
$P \IPOtr{C[-]} P'$.
\end{definition}

The next two definitions are needed to ensure that $\bsbis \subseteq
\bisl$.

\begin{definition}\label{def:PredStable}
Let $\mathcal R$ be a relation and let $\mathcal P(X,Y)$ be a
predicate on processes. We say that $\mathcal P(X,Y)$ is
\emph{stable under $\mathcal R$} if whenever $P \mathcal R Q$ and
$\mathcal P(P,P')$ there exists $Q'$ such that $\mathcal
P(Q,Q')$ and $P' \mathcal R Q'$.
\end{definition}

For example, the predicates in Fig.~\ref{PredMA} and
Fig.~\ref{fig:arrayCCS} are stable under $\bsbis$.

\begin{definition}\label{def:Lstable}
Let $\mathcal R$ be a relation and let $C[-]$ be a label. We say
that $C[-]$ is \emph{stable under $\mathcal R$} if the predicate
$\mathcal P(X,Y)= X \IPOtr{C[-]} Y$ is stable under $\mathcal R$.
\end{definition}

We can finally state a first correspondence result.

\begin{proposition} \label{Prop:BSB=LB} Let us consider an RS with
  redex RPOs, a set $O$ of contextual barbs and a set $L$ of
  labels. If $L$ is $O$-capturing and its labels are stable under
  $\bsbis$, then $\bsbis$ coincides with $\bisl$.
\end{proposition}

\begin{proof}
In order to prove that $\bsbis \subseteq \bisl$, we show that
$\mathcal R=\{(P,Q)$ s.t. $P\bsbis Q \}$ is an
$L$-bisimulation.

Suppose that $P \IPOtr{C[-]}P'$. We have two cases: either $C[-]\in
L$ or $C[-]\notin L$. If $C[-]\in L$, then $C[-]$ is stable under
$\bsbis$ and thus, since $P\bsbis Q$, $Q \IPOtr{C[-]}Q'$ and
$P'\bsbis Q'$. For the case that $C[-]\notin L$, it is enough to
note that, since $P \IPOtr{C[-]}P'$, then $C[P]\react P'$. Since
$P\bsbis Q$, then $C[Q]\react Q'$ and $P' \bsbis Q'$.

Now we show that $\mathcal R=\{(P,Q)$ s.t. $P\bisl Q \}$
is a barbed semi-saturated bisimulation (i.e., $\bisl \subseteq
\bssbis$) and thus, since the RS has redex IPOs, by
Proposition \ref{PropBSS=BS} it follows that $\bisl \subseteq
\bsbis$.

At first, we note that, since $O$ is a set of contextual barbs,
in order to show that $\mathcal R$ satisfies the first condition of
Definition \ref{def:BSSbis} it suffices to show that $P\barb o$
implies $Q\barb o$. Since $L$ is $O$-capturing,  if $P\barb o$
then there is a label $C[-]\in L$ such that $P\barb o$
if and only if $P\IPOtr{C[-]}$. Since $P\bisl Q$, then also
$Q\IPOtr{C[-]}$ and $Q\barb o$.

In order to prove the second condition of Definition
\ref{def:BSSbis}, it is enough to note that if $P\IPOtr{C[-]}P'$
then, for either $C[-]\in L$ or $C[-]\notin L$,
$C[Q]\react Q'$ with $P'\bisl Q'$.
\end{proof}

As a corollary of the previous definition, we obtain the following property
that allows to check whenever IPO-bisimilarity coincides with barbed
saturated one.

\begin{lemma}
  Let us consider an RS with redex IPOs and a set $O$ of contextual
  barbs. If the set of all labels is $O$-capturing and each label is
  stable under $\bsbis$, then $\sim^I$ coincides with $\bsbis$.
\end{lemma}


\section{$L$-Bisimilarity for Mobile Ambients}\label{sec:LbisMA}
This section proposes a new labelled characterization of the
reduction barbed congruence for MAs, presented in Section~\ref{sec:MA}.
In particular, by using the ITS $M_I$ (also
in Section~\ref{sec:MA}) we define an $L$-bisimilarity
that captures barbed saturated bisimilarity for MAs,
coinciding with reduction barbed congruence.


\begin{proposition}[see~\cite{BGMFOSSACS09}, Theorem~3] 
\label{RBC=BSB}
  Reduction barbed congruence over MAs $\sim^{MA}$ coincides with
  barbed saturated bisimilarity $\bsbis_{\scriptscriptstyle{M}}$.
\end{proposition}

As shown in Section~\ref{L=BS}, we can characterize barbed saturated bisimilarity
on a set of contextual barbs $O$ through the IPO transition system and a set of labels $L$.
In particular, as required by Proposition \ref{Prop:BSB=LB}, the set $L$
must be $O$-capturing and each $C[-] \in L$ must be stable under the barbed saturated bisimilarity.

We denote by $O_{\scriptscriptstyle{M}}$ the set of barbs of MAs,
recalling that MAs barbs are contextual barbs~\cite{BGMFOSSACS09}.

\begin{proposition}[see~\cite{BGMFOSSACS09}, Proposition~6] 
  \label{MAsCB}
$O_{\scriptscriptstyle{M}}$ is a set of contextual barbs.
\end{proposition}

Therefore, we can characterize reduction barbed congruence over MAs by instantiating
Definitions~\ref{def:L1bis} with the $\ilts$ $M_I$ and a set $L$ of labels having the two properties
said above.

First of all, we find some labels of $M_I$ that capture the barbs of MAs.
This ensures that the checking of barbs of the barbed saturated bisimilarity is
done in the $L$-bisimilarity by the first condition of its definition.
It is easy to note that a MAs process $P$ observes a unrestricted ambient $n$ at top-level,
in symbols $P \barb n$, if and only if it can execute a transition labelled with
$-|open \hspace{1 mm} n.T_1$ or with $-|m[in \hspace{1mm} n.T_1|T_2]$.
Therefore, $L$ is $O_{\scriptscriptstyle{M}}$-capturing if it contains at least one kind of these labels.
We choose to consider labels of the first type, that is, having
the shape  $-|open \hspace{1mm} n.T_1$, for $n$ ambient name and $T_1$ pure process.

It is possible to prove that these labels are stable under
$\bsbis_{\scriptscriptstyle{M}}$. Therefore,
if we consider the set $L$ defined below, we obtain an $L$-bisimilarity for MAs
that is able to characterize $\bsbis_{\scriptscriptstyle{M}}$.



\begin{proposition} \label{LMAO-capturing}\label{def:LMAs}
Let $L_M$ be the set of all labels of the $\ilts$ $M_I$ having
the shape  $-|open \hspace{1mm} n.T_1$, for $n$ ambient name and $T_1$ pure process.
Then, $L_M$ is $O_{\scriptscriptstyle{M}}$-capturing.
\end{proposition}

\begin{proof}
We have to show that for each barb $n \in O_{\scriptscriptstyle{M}}$ there exists a label $C[-] \in L_M$
such that for each process $P$, $P \barb n$ if and only if $P \tr{C[-]}_{M_I} P'$.

It is easy to note that, given a barb $n \in
O_{\scriptscriptstyle{M}}$, we have that for each process $P$, $P
\barb n$ if and only if $P \tr{-|open \hspace{1mm} n.T_1}_{M_I} P'$,
with $T_1$ pure process. Since we know that $L_M$ contains all
labels having the shape  $-|open \hspace{1mm} n.T_1$, for $n$
ambient name and $T_1$ pure process, we can conclude that $L_M$ is
$O_{\scriptscriptstyle{M}}$-capturing.
\end{proof}

Now, in order to prove that each $C[-] \in L_M$ is stable under
$\bsbis_{\scriptscriptstyle{M}}$, we exploit a predicate such that
it is stable under $\bsbis_{\scriptscriptstyle{M}}$ and equivalent
to the one of Definition \ref{def:Lstable}.

\begin{lemma} \label{lemma:PredMA}
Let $\mathcal P^{-|open \hspace{1mm} n.T_1}(X,Y)$ be the binary predicate
on MAs processes shown in Fig. \ref{PredMA}, for $n$ ambient name
and $T_1$ pure process.
Then, $\mathcal P^{-|open \hspace{1mm} n.T_1}(X,Y)$ is stable under $\bsbis_{\scriptscriptstyle{M}}$
and for each $P$ and $P'$, $\mathcal P^{-|open \hspace{1mm} n.T_1}(P,P')$ if and only if
$P \tr{-|open \hspace{1mm} n.T_1}_{M_I} P'$.
\end{lemma}

\begin{proof}
We begin by proving that the predicate $\mathcal P^{-|open \hspace{1mm} n.T_1}(X,Y)$
is stable under $\bsbis_{\scriptscriptstyle{M}}$.

Assume that $P \bsbis_{\scriptscriptstyle{M}} Q$ and $\mathcal P^{-|open \hspace{1mm} n.T_1}(P,P')$ holds.
Since $\mathcal P^{-|open \hspace{1mm} n.T_1}(P,P')$ holds, then
there exists a process $P''$ and an ambient $m$ fresh for $P$ and $Q$,
such that $C'[P] \react P''$, $P'' \barb m$,
$P'' \react P'$ and $P' \not \barb m$, with
$C'[-] = -|open \hspace{1mm} n.(m[\nil]|open \hspace{1mm} m.T_1)$.

Since $C'[P] \react P''$ and $P \bsbis_{\scriptscriptstyle{M}} Q$, then
$C'[Q] \react Q''$ and $P'' \bsbis_{\scriptscriptstyle{M}} Q''$.
Therefore, it is obvious that also $Q'' \barb m$.
Now, we know that $P'' \react P'$, hence we can say
that $Q'' \react Q'$ and $P' \bsbis_{\scriptscriptstyle{M}} Q'$.
From this follows that, since $P' \not \barb m$, then also $Q' \not \barb m$.
So, we can conclude that $\mathcal P^{-|open \hspace{1mm} n.T_1}(Q,Q')$ holds, hence $\mathcal P^{-|open \hspace{1mm} n.T_1}(X,Y)$
is stable under $\mathcal R$.

Now we show that for each $P$ and $P'$, $\mathcal P^{-|open \hspace{1mm} n.T_1}(P,P')$ iff
$P \tr{-|open \hspace{1mm} n.T_1}_{M_I} P'$.

Assume that $\mathcal P^{-|open \hspace{1mm} n.T_1}(P,P')$ holds.
This means that there exists a process $P''$ and an ambient $m$ fresh for $P$,
such that $C'[P] \react P''$, $P'' \barb m$,
$P'' \react P'$ and $P' \not \barb m$, with
$C'[-] = -|open \hspace{1mm} n.(m[\nil]|open \hspace{1mm} m.T_1)$.
The fact that $C'[P] \react P''$ and $P'' \barb m$ means that
the capability $open \hspace{1mm} n$ has been executed, hence
there must be a unrestricted ambient $n$ at top-level of $P$, i.e.,
$P \equiv (\nu A) (n[P_1]|P_2)$ and $n \not \in A$.
From this follows that $P'' = (\nu A)(P_1|P_2)|m[\nil]|open \hspace{1 mm} m.T_1$,
and since $P' \not \barb m$, then $P' \equiv (\nu A) (P_1|P_2)|T_1$.
Moreover, by knowing that $P = (\nu A) (n[P_1]|P_2)$ and $n \not \in A$,
we can conclude that $P \tr{-|open \hspace{1mm} n.T_1}_{M_I} P'$.

Assume that $P \tr{-|open \hspace{1mm} n.T_1} P'$. This means that
$P \equiv Q$, where $Q = (\nu A) (n[P_1]|P_2)$, $n \not \in A $ and
$P'= (\nu A) (P_1|P_2)|T_1$. We consider the context $C'[-] = -|open
\hspace{1mm} n.(m[\nil]|open \hspace{1mm} m.T_1)$ with $m\not \in
fn(P)$. It is easy to note that $C'[Q] \react P''$ s.t. $P'' = (\nu
A) (P_1|P_2)|m[\nil]|open \hspace{1 mm} m.T_1$ and $P'' \barb m$.
Therefore, since $C'[P] \equiv C'[Q]$, we also have that $C'[P]
\react P''$. Now, we can note that $P'' \react P'$ and, since $m$ is
fresh for $P$, $P' \not \barb m$.
\end{proof}

\begin{figure}[!t]
\begin{tabular}{p{3 cm} p{12.1 cm}}
\hline
\\
$\mathcal P^{-|open \hspace{1mm} n.T_1}(X,Y)$ & $\exists P''$ and $m \not \in fn(X)$ s.t. $P'' \barb m,
C'[X] \react P'' \react Y$ and $Y \not \barb m$ \\
& with $C'[-] = -|open \hspace{1mm} n.(m[\nil]|open \hspace{1mm} m.T_1)$\\ \\
\hline
\end{tabular}
\caption{Predicate for the label $-|open \hspace{1mm} n.T_1$.}
\label{PredMA}
\end{figure}

\begin{proposition} \label{LMAStable}
All labels in $L_M$ are stable under $\bsbis_{\scriptscriptstyle{M}}$.
\end{proposition}

The proof of the proposition above trivially follows from Lemma \ref{lemma:PredMA}.

We finally introduce the main characterization proposition.

\begin{proposition}
$\bsbis_{\scriptscriptstyle{M}} = \sim^{L_M}$.
\end{proposition}

\begin{proof}
  First of all, by Proposition \ref{MAsCB}, we know that MAs barbs are
  contextual.  Moreover, by Propositions \ref{LMAO-capturing} and
  \ref{LMAStable}, we know that $L$ is
  $O_{\scriptscriptstyle{M}}$-capturing and it contains only labels
  that are stable under $\bsbis_{\scriptscriptstyle{M}}$.  Therefore,
  thanks to Proposition \ref{Prop:BSB=LB}, we can conclude that
  $\bsbis_{\scriptscriptstyle{M}}= \sim^{L_M}$.
\end{proof}

The $L$-bisimilarity $\sim^{L_{M}}$ presented above is not the only
one which is able to characterize barbed saturated bisimilarity
$\bsbis_{\scriptscriptstyle{M}}$. For example, as said before, we
can choose to consider all labels of the shape $-|m[in \hspace{1mm}
n.T_1|T_2]$: besides being able to capture MAs barbs, they are
also stable under $\bsbis_{\scriptscriptstyle{M}}$. However,
generally, we can consider the sets $L$ containing at least all
the labels of the shape $-|open \hspace{1mm} n.T_1$ or $-|m[in
\hspace{1mm} n.T_1|T_2]$ to capture barbs, and other labels of $M_I$
that are stable under $\bsbis_{\scriptscriptstyle{M}}$, i.e., labels
such that it is possible to define a predicate analogous to the one
we defined for the labels $-|open \hspace{1mm} n.T_1$.


\section{$L$-Bisimilarity for (Asynchronous) CCS} \label{sec:L-BisCCS}

Section \ref{sec:CCS} has shown that IPO-bisimilarity coincides with
the ordinary bisimilarity of CCS ($\sim^{CCS}$), while it is
strictly contained in asynchronous bisimilarity. In this section, we
first show that $L$-bisimilarity generalizes both cases and then we
prove that these also coincide with their barbed saturated
bisimilarities.

\paragraph{$L$-Bisimilarity for Asynchronous CCS.}
In asynchronous bisimulation (Definition \ref{def:Abis}),
transitions labelled with $\tau$ and $\bar{a}$ (corresponding to $-$
and $-|a.T_1$ in $A_I$, respectively) must be matched by transitions
with the same labels. Moreover, when $P \tr{a}P'$ (corresponding to
$P \tr{-|\bar{a}}P'$ in $A_I$) then either $Q\tr{a}Q'$ and $P'\,
\mathcal R\,Q'$ or $Q\tr{\tau}Q'$ and $P'\, \mathcal R \, Q' |
\bar{a}$. This is equivalent to require that $Q | \bar{a} \react Q'$
and $P'\, \mathcal R \, Q'$.
Thus, in order to characterize $\sim^A$ as $L$-bisimilarity, it
suffices to choose as $L$ the set of labels corresponding to $\tau$
and $\bar{a}$.


\begin{proposition}
  Let $L_{A}$ be the set containing the labels of the ITS $A_I$ of the
  shape $-$ and $- |a.T_1$, for $a$ channel name and $T_1$ pure
  process.  Then, $\sim^{L_{A}}=\sim^{A}$.
\end{proposition}

\paragraph{$L$-Bisimilarity for CCS.} Since IPO-bisimilarity coincides
with $\sim^{CCS}$, in order to characterize $\sim^{CCS}$ as
$L$-bisimilarity, it is enough to include all the IPO-labels into
$L$.


\begin{proposition}
Let $L_{CCS}$ be the set containing all the labels of the ITS $C_I$.
Then, $\sim^{L_{CCS}}=\sim^{CCS}$.
\end{proposition}

\paragraph{From $L$-Bisimilarity to Barbed Saturated Bisimilarity.}
It is important to note that the choice of $L_{CCS}$ and $L_A$ is
not arbitrary. Indeed, in both cases, $\sim^{L_{CCS}}$ and
$\sim^{L_A}$ coincide with barbed saturated bisimilarities. This is
not a new result, but it is interesting to see that it can be easily
proved by following the same approach that we have used for MAs in
Section~\ref{sec:LbisMA}.

For the synchronous case, barbs are defined as $P\downarrow a$ if and
only if
$P\tr{a}Q$ and $P\downarrow \bar{a}$ if and only if $P\tr{\bar{a}}Q$. Since
$L_{CCS}$ contains the labels $-|\bar{a}.T_1$ and $-|a.T_1$
(corresponding to $a$ and $\bar{a}$ in the ordinary LTS), then
$L_{CCS}$ is barb capturing.

It is also easy to see that the barbs are contextual. Then, in order
to use Proposition \ref{Prop:BSB=LB}, we only have to prove that all
the labels in $L_{CCS}$ are stable under barbed congruence.
Analogously to MAs, we define some additional predicates. These are
shown in Fig. \ref{fig:arrayCCS}. It is easy to see that for each
label $C[-]$, $X\tr{C[-]}Y$ in $C_I$ if and only if $\mathcal
P^{C[-]}(X,Y)$. It is also easy to show that all of them are stable
under $\bsbis$.

\begin{figure}
\begin{tabular}{p{3 cm} p{12.1 cm}}
\hline
\\
  $ \mathcal P^{-|\bar{a}.T_1}(X,Y)$ & $\exists P'$ and $i\notin fn(X)$  s.t.
  $P'\downarrow \bar{i}$ and $X|\bar{a}.(\bar{i}|T_1) | i \react
  P' \react Y$ \\
  $\mathcal P^{-|a.T_1}(X,Y)$ & $\exists P'$ and $i\notin fn(X)$  s.t.
  $P'\downarrow \bar{i}$ and $X|a.(\bar{i}|T_1) | i \react
  P' \react Y$ \\
  $\mathcal P^-(X,Y)$ & $X \react Y$ \\ \\
\hline
\end{tabular}
\caption{Predicates for CCS}\label{fig:arrayCCS}
\end{figure}

For the asynchronous case, recall that $L_A$ only contains labels of
the form $-$ and $-|a.T_1$ (corresponding to labels $\tau$ and
$\bar{a}$ in the ordinary LTS). Since only output barbs
$\downarrow_{\bar{a}}$ are defined, then $L_A$ is barb capturing. In
order to prove that each label in $L_A$ is stable under $\bsbis$ we
can use for $-$ and $-|a.T_1$ the predicates that we have used in
the synchronous case (Fig.~\ref{fig:arrayCCS}).

It is worth noting that labels of the form $-|\bar{a}$ are not
stable under $\bsbis$. Indeed, we cannot adopt the predicate used in
the synchronous case (the first in Fig.~\ref{fig:arrayCCS}), since
outputs have no continuation in ACCS.


\section{Conclusions and future work}
\label{sec:Concl}

The paper introduces a novel behavioural equivalence for RSs, namely,
$L$-bisimulation: a flexible tool, parametric with respect to a set of
labels $L$. The associated bisimilarity is proved to be a
congruence, and it is shown to be intermediate between the standard
IPO and saturated semantics for RSs: indeed, it is able to recover
both of them, by simply varying the set of labels $L$. More
importantly, also the more expressive barbed saturated semantics can
be recast, as long as the set $L$ satisfies suitable conditions.

As for any newly proposed semantics, its expressiveness and ease of
use have to be tested against suitable case studies. We thus
considered a recently introduced IPO transition system for
MAs, and two other IPO transition systems for CCS and its asynchronous
variant. We show that in all those cases, for a right choice of $L$,
$L$-bisimilarity precisely captures the standard semantics for the calculus at hand.

We can foresee three immediate extensions of our work. First of all,
we would like to precisely understand the notion of IPO-closedness,
which is required for the set of labels $L$, in order for
$L$-bisimilarity to be a congruence. It would be important to
establish suitable and more manageable conditions under which a set of
arrows of a given category satisfies that property, especially for
those RSs where IPOs have an inductive presentation (such as for those
induced by the borrowed context mechanism).

Moreover, we would like to further elaborate on the connection between
$L$-bisimilarity and barbed semantics, moving beyond the preliminary
results presented in Section~\ref{L=BS}. As a start, in order to
establish conditions ensuring that barbs satisfy the pivotal property
of being contextual; and, more to the point, for checking whenever a
set of labels is barb capturing and contains only labels stable under
barbed saturated bisimilarity.
As far as the specific MAs case study
is concerned, most of the IPO labels occurring in our transition
system are indeed stable, i.e., the relative labelled transitions
can be characterized by a
predicate which is stable under the barbed saturated bisimilarity.
The only labels that are not stable are the ones
of the shape $-|m[P]$ and $m[-|P]$ of the rule \textsc{InAmb}
and \textsc{OutAmb}, respectively.
It seems intriguing that those same labels required the
introduction of so-called Honda-Tokoro inference rules
in~\cite{PawelAmbient} for capturing the reduction barbed congruence
by means of standard bisimilarity.

Finally, we remark that so far in our methodology the choice of the
``right'' set $L$, as well as the identification of a meaningful set
of barbs, is left to the ingenuity of the researcher. We would like to
devise a general theory that relying only on the syntax of the
calculus at hand and on the associated reduction semantics might allow
to automatically derive either a suitable family of barbs or some kind
of basic set of observations, along the lines of the proposals
in~\cite{HY,PawBarbs,RathkeS09}.

\paragraph{\textbf Acknowledgements.}
We are indebted to the anonymous referees for their useful remarks, which helped
us in improving the overall presentation of the paper.

\bibliographystyle{eptcs}
\bibliography{biblio}
\end{document}

%% file: macro.tex
\newcommand{\bisl}{\sim^{L}}

\newcommand{\nil}{\mathbf{0}}

     \def\tr#1{\stackrel{#1}{\to}}
     \def\IPOtr#1{\stackrel{#1}{\rightarrowfill_{IPO}}}
     \def\FULLtr#1{\stackrel{#1\ \ }{\rightarrowfill_{SAT}}}

  \newcommand{\proofend}{\mbox{$\Box$}}

\newcommand{\SQUR}[5][]{\ensuremath{{{}_{#1{#2}}}_{\to#1{#4}\smash{{}^{\nearrow}}}^{{}_{\nearrow}#1{#3}\to#1{#5}}}}
\newcommand{\squr}[5][r]{%
  \ifthenelse{\equal{#1}{r}}%
  {\SQUR{#2}{#3}{#4}{#5}}%
  {\ensuremath{\mathreflect{\SQUR[\mathreflect]{#2}{#3}{#4}{#5}}}}%
}

\let\oldrightarrow\to
\renewcommand{\to}[1][]{%
  \ifthenelse{\equal{#1}{}}%
  {\oldrightarrow}%
  {\ensuremath{\xrightarrow{{}_{#1}}}}%
}

\let\oldleftarrow\gets
\renewcommand{\gets}[1][]{%
  \ifthenelse{\equal{#1}{}}%
  {\oldleftarrow}%
  {\ensuremath{\xleftarrow{{}_{#1}}}}%
}

\makeatletter
\def \rightarrowfill{\m@th\mathord{\smash-}\mkern-6mu%
  \cleaders\hbox{$\mkern-2mu\mathord{\smash-}\mkern-2mu$}\hfill
  \mkern-6mu\mathord\rightarrow}
\makeatother \makeatletter
\def \Rightarrowfill{\m@th\mathord{\smash-}\mkern-6mu%
  \cleaders\hbox{$\mkern-2mu\mathord{\smash-}\mkern-2mu$}\hfill
  \mkern-6mu\mathord\Rightarrow}
\makeatother \makeatletter
\def \mapstofill{\m@th\mathord{\smash-}\mkern-6mu%
  \cleaders\hbox{$\mkern-2mu\mathord{\smash-}\mkern-2mu$}\hfill
  \mkern-6mu\mathord\longmapsto}
\makeatother

\newcommand{\Cat}[1]{{\bf #1}}

\newcommand{\Rules}[1]{\mathfrak{#1}}

\newcommand{\barb}[1]{\downarrow_{#1}}

\newcommand{\bsbis}{\sim^{BS}}

\newcommand{\bssbis}{\sim^{BSS}}

\newcommand{\ilts}{\textsc{ITS}}

\def\tr#1{\stackrel{#1}{\rightarrowfill}}
\def\IPOtr#1{\stackrel{#1}{\rightarrow_{IPO}}}
\def\FULLtr#1{\stackrel{#1\ \ }{\rightarrow_{SAT}}}
\newcommand{\react}{\rightsquigarrow}

\renewcommand{\>}{\rangle}

\def\notaF#1{ }

\newcommand{\short}[1]{}

\newcommand{\longrew}[1]{\ensuremath{\!\xymatrix{ {} \ar@{=>}[r]^{#1} & }\!}}

\def\tr#1{\stackrel{#1}{\rightarrowfill}}

\makeatletter
\def \rightarrowfill{\m@th\mathord{\smash-}\mkern-6mu%
  \cleaders\hbox{$\mkern-2mu\mathord{\smash-}\mkern-2mu$}\hfill
  \mkern-6mu\mathord\rightarrow}
\makeatother \makeatletter
\def \Rightarrowfill{\m@th\mathord{\smash-}\mkern-6mu%
  \cleaders\hbox{$\mkern-2mu\mathord{\smash-}\mkern-2mu$}\hfill
  \mkern-6mu\mathord\Rightarrow}
\makeatother \makeatletter
\def \mapstofill{\m@th\mathord{\smash-}\mkern-6mu%
  \cleaders\hbox{$\mkern-2mu\mathord{\smash-}\mkern-2mu$}\hfill
  \mkern-6mu\mathord\longmapsto}
\makeatother

\newcommand{\cellp}[5]{\xymatrix@C=#5cm{{#1} \ar[r]^{#2} \ar@{}[r]_{#3} & {#4}}}
\newtheorem{definition}{Definition}{}
{}
\newtheorem{lemma}{Lemma}{}
{}
\newtheorem{proposition}{Proposition}{}
{}
{}